\documentclass[11pt]{article}
 
\usepackage{amsmath,amsthm} 
\usepackage{amsfonts} 
\usepackage{amssymb}
\usepackage{graphicx} 

\newcommand{\Lam}{\Lambda} 
 
\newcommand{\al}{\alpha}
\newcommand{\be}{\beta}

\newtheorem{theorem}{Theorem}
\newtheorem{definition}{Definition}

\theoremstyle{remark}

\numberwithin{equation}{section}

\newcommand{\lf}{\mathcal{L}}
\newcommand{\bmm}{\mathbf{B}}
\newcommand{\lmm}{\mathbf{L}}
\newcommand{\umm}{\mathbf{U}}
\newcommand{\dmm}{\mathbf{D}}
\newcommand{\set}[1]{\left\{#1\right\}}

\newcommand{\R}{\mathcal{R}}

\newcommand{\Z}{\mathbb{Z}}

\newcommand{\mmod}{\mathcal{M}}

\newcommand{\vsub}{\mathcal{V}}

\newcommand{\nP}{\hat{P}}
\newcommand{\nQ}{\hat{Q}}

\newcommand{\secref}[1]{Section \ref{#1}}
\newcommand{\thmref}[1]{Theorem \ref{#1}}

\newcommand{\figref}[1]{Figure \ref{#1}}

\newcommand{\bra}[1]{\langle{#1}|}
\newcommand{\ket}[1]{|#1\rangle}
\newcommand{\bran}[1]{\langle{#1}}

\newcommand{\lfp}{\lf }

\newcommand{\db}{\mathbf{d}}
\newcommand{\eb}{\mathbf{e}}

\begin{document}

 \setcounter{page}{1}

\title{Bi-orthogonal Polynomial Sequences and the Asymmetric Simple Exclusion Process}

\author{R. Brak\thanks{rb1@unimelb.edu.au}\hspace{1ex}   and W. Moore   
    \vspace{0.15 in} \\
         Department of Mathematics,\\
         The University of Melbourne\\
         Parkville,  Victoria 3052,\\
         Australia\\
 }

\maketitle

\begin{abstract} 
 We state the  diffusion algebra equations of the stationary state of the three parameter ($\alpha$, $\beta$ and $q$) Asymmetric Simple Exclusion Process  as a linear functional $\lf$, acting on a tensor algebra. 
 From $\lf$ we construct a  pair of  sequences, $\set{P_n}$ and $\set{Q_m}$, of  monic polynomials which are bi-orthogonal, that is,  they satisfy \mbox{$\lf(P_n  Q_m )=\Lambda_n\delta_{n,m}$} (where $\Lambda_n$ is a scalar). The uniqueness and existence of the pair of sequences arises from the determinant of the bi-moment   matrix  whose elements satisfy a pair of  $q$-recurrence relations.   The determinant is evaluated using an LDU-decomposition.
 If the linear functional is represented as an inner product, $\lf(\cdot)=\bra{W} \cdot \ket{V}$ then the action of the polynomials $Q_n$ on the boundary vector $\ket{V}$ generate  a basis  $\ket{V_n}=Q_n\ket{V}$ whose orthogonal dual vectors are given by the action of $P_n$ on the dual boundary vector $\bra{W}$, that is $\bra{W_n}=\bra{W}P_n$.
  This basis  gives the representation of the algebra which is associated with the  Al-Salam-Chihara polynomials obtained by Sasamoto.

\end{abstract}

\vfill
\newpage

\section{Introduction}\label{sec1}

The   Asymmetric Simple Exclusion Process  (ASEP) is a continuous time Markov process defined by particles hopping along a line of  $L$ sites -- see \figref{fig:hop}.  Particles hop on to the line on the left with rate $\alpha$, off at the right with rate $\beta$ and   they hop to neighbouring sites to the left with  rate $q$ and rate one to the right with the constraint that only one particle can occupy a site. 
 \begin{figure}[ht]
\begin{center}
\includegraphics{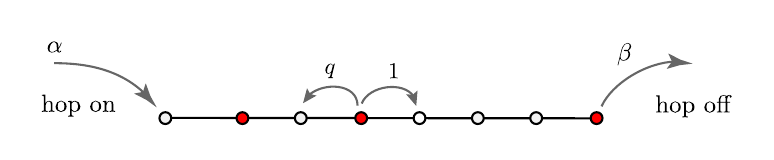}  
\caption{Three parameter ASEP hopping model}\label{fig:hop}
\end{center}
\end{figure}
  The problem of computing the  stationary probability distribution was solved by Derrida et.\ al.\ \cite{derrida97}    with the introduction of the  matrix product Ansatz (see below) which provides an algebraic method of computing the stationary distribution.  A recent review of the Asymmetric Exclusion Process may be found in  Blythe  and  Evans \cite{Blythe:2007uq}.   
 
The matrix product Ansatz expresses the stationary distribution of a given state as an inner product on a certain quotient ring of matrices. This ring is generated by two matrices $D$ and $E$ which satisfy the relation
\[
DE-qED-D-E=0\,.
\]
The inner product $\bra{W}\cdot \ket{V}$ is then defined by two vectors $\bra{W}$ and $\ket{V}$ (which we will refer to as boundary vectors) which satisfy $(\beta  D-1 )\ket{V}=0$ and $\bra{W}(\alpha E-1)=0$.

Rather than using $D$ and $E$ the algebra is  simplified by working with the shifted variables,   
\begin{subequations}
\begin{align}\label{eq_shift}
	d&=q'D-1\,,\\
	e&=q'E-1\,,
\end{align}	
\end{subequations} 
where $q'=1-q$. In these variables the above relation takes on the well known form (see for example, \cite{Lazarescu:2014ab}),
\[
de-q\, ed=q'\,.
\]

Computing representations of the $d$ and $e$ matrices fall into  natural cases. The case with  $\alpha $ and $\beta $ non-zero but $q=0$ we will refer to as the two parameter model and the case    $\alpha$, $\beta $ and $q$ non-zero as the three parameter model. There is also a five parameter model which has hopping off on the left with rate $\gamma$ and on on the right with rate $\delta$ which we do not directly address in this paper.  

The paper by Derrida et al \cite{derrida97} originally found three different representations for the two parameter case.   Representations of the three parameter model can be found in \cite{Sasamoto:1999aa} (and references therein) and for the five parameter model in \cite{UCHIYAMA:2003aa} (and references therein). 

If the matrices associated with a given representation  have sufficiently simple structure (eg.\ bi- or tri-diagonal) then they can be usefully interpreted as transfer  matrices  for  lattice path models \cite{Brak2004vf}. This leads to   combinatorial methods for computing  the inner product. 
 
Each matrix representation is  associated with a basis for the vector space upon which the matrices act. A very well known basis is the set $\ket{n}$, $n\ge0$, generated by the  action of $e^n$ on a vacuum vector $\ket{0}$  defined by $d\ket{0}=0$. 
For the three parameter model  this basis and its dual give  a representation in which the components of the boundary vectors are related to $q$-binomial coefficients and the tri-diagonal matrix $d+e$  gives a three term recurrence related to  $q$-Hermite polynomials \cite{UCHIYAMA:2003aa}.

The primary objective of this paper  is the basis associated with the three parameter model representation obtained by Sasamoto \cite{Sasamoto:1999aa} where   the tri-diagonal matrix $d+e$  gives a three term recurrence related to   the Al-Salam-Chihara polynomials \cite{gasper90}. 
We show that this basis   is associated with a pair of distinct sequences, $\set{P_n}$ and $\set{Q_m}$, of  polynomials. 
 The  polynomials $Q_n$ generate the  basis when acting on the  boundary vector $\ket{V}$  and the orthogonal dual vectors  are generated by the  polynomials $P_n$   acting on the  dual boundary vector $\bra{W}$.  
 Thus the basis is the set of vectors  $\ket{V_n}=Q_n\ket{V}$ and the orthogonal dual basis  set is $\bra{W_n}=\bra{W}P_n$.  Since the basis is generated by the boundary vector $\ket{V}$ and its dual $\bra{W}$ we will refer to this as the ``boundary basis''.
 
We show the  two polynomial sequences  are bi-orthogonal with respect to a certain linear functional $\lf$, that is       
$\lf\bigl(P_n   Q_m \bigr)=\Lambda_n\delta_{n,m}\,,$
where $\Lambda_n$ is a scalar.  For convenience the pair $\set{P_n}$ and $\set{Q_m}$ will be referred to as a bi-orthogonal pair of polynomial sequences, or BiOPS.

 The uniqueness and existence of the BiOPS arises from the determinant of the bi-moment   matrix whose elements are given by $B_{n,m}=\lf (d^n e^m )$. These elements  satisfy a pair of  $q$-recurrence relations. 
Unlike traditional orthogonal polynomials defined by Favard's theorem (see \cite{favard:1935qy} or \cite{Ismail:2005aa}) (ie.\ they satisfy a three-term recurrence relation), the BiOPS satisfy first order (uncoupled) $q$-recurrence relations.    We show that the BiOPS are intimately associated with the  decomposition of the bi-moment matrix into upper and lower triangular matrices. In fact the polynomial coefficients, when written in their own basis,  are the matrix elements of the lower (for $P_n$) and upper (for $Q_n$) matrices -- see equation \eqref{eq_polexpn}.

%

\section{The Algebra}
\label{sec_mods}


In this section we set up the tensor algebra used to represent the ASEP \cite{Crampe:2014ab}. Let $\R$ be the ring of integer coefficient commutative polynomials, $\Z[\al,\be,q]$ and $\mmod$   the  $\R$-module
\begin{equation}\label{eq_mmod}
\mmod= \bigoplus_{n\ge0}  \vsub_2^{\otimes n}
\end{equation}
where $\vsub_2$ is a free rank two $\R$-module with generators $ d,e$. Here $\vsub^0$ denotes the ring $\R$  of the module and $\vsub_2^{\otimes n}=\vsub_2\otimes\vsub_2\otimes\cdots\otimes\vsub_2$ ($n$ factors). 
 
The homogeneous submodule $\vsub_2^{\otimes n}$, of degree $n$, is generated by the  standard  monomial basis elements $e_{i_1} \otimes  e_{i_2} \cdots  \otimes e_{i_n} $ where $e_i\in\set{d,e}$. 
 For brevity  we will frequently omit the tensor product symbol, thus $d^m e^n$ denotes  $d^{ \otimes m}\otimes e^{ \otimes n}$ etc. 

We use the  three parameter version of the  original matrix Ansatz algebra equations of Derrida et al \cite{derrida97} as modified in \cite{corteel:2010rt}. The latter form  allows for arbitrary monomial pre- and post-factors ($u$ and $v$ in the equations below). 
The original algebra was stated   in terms of matrices and vectors. Here we give a slightly more abstract version by using  a  linear functional and use the shifted variables $d$ and $e$ rather than $D$ and $E$.
\begin{definition} \label{thm_cw}

Let $u,v$ be any monomial basis elements of $\mmod$. 
The \mbox{$\R$-module}  homomorphism  $\lf :\,\mmod\to \R$ is defined by the following equations:
\begin{subequations}\label{eq_cw}
\begin{align}
 \lf(u\otimes (d\otimes  e - q\, e\otimes  d -q' )\otimes v)&= 0\label{eq_cw_a}\\
 \lf(u\otimes (d - b ) )&=0\label{eq_cw_b}\\
 \lf(  (e-a)\otimes v)&=0\label{eq_cw_c}
\end{align}
\end{subequations}
where $a=  q'/\al-1 $, $b=q' /\be-1 $, with $\lf(1)=1$  and extended linearly to other elements of $\mmod$.

\end{definition}

 The reasons for the slightly more abstract linear functional formulation are as follows. The primary reason is because this is how traditional three-term recurrence  polynomial orthogonality can be formulated  (see Favard's theorem \cite{favard:1935qy}). 
 This in turn allows for a direct combinatorial  construction of   orthogonality  \cite{Viennot1985ah} without going via any integral representations. 
 It also allows for other representations of the linear functional such as via double integral measures \cite{M.-Bertola1:2002aa}  or via inner products as was done in the original Derrida et.\ al.\ paper \cite{derrida97}.

The matrix product Ansatz of  \cite{derrida97}  for the stationary state can now be (trivially) restated   using  the   linear functional $\lf$.
\begin{theorem}[ Derrida, Evans,  Hakin  and  Pasquier \cite{derrida97}]
The stationary state probability distribution, $f(\tau)$, of the two parameter ASEP for the system in state $\tau=(\tau_1,\cdots,\tau_L)$, is given by
\begin{align}
f(\tau) =&\frac{1}{Z_L}\, \lf\left(\prod_{i=1}^L(\tau_i d+(1-\tau_i)e)\right)\, \label{eq_prob}
\intertext{where}
Z_L  =&\lf\biggl((d+ e)^L\biggr)\,\label{norm_eq}
\end{align}
and $\tau_i=1$ if site $i$ is occupied and zero otherwise.
\end{theorem}
%

\section{Bi-Orthogonal Pair of Polynomial  Sequences}
\label{biops}

Consider the pair of sequences, 
\begin{equation}
 \set{P_n(d)}_{n\ge0},\qquad \text{}\qquad  \set{Q_n(e)}_{n\ge0} 
\end{equation}
of  monic polynomials (where $P_n$ and $Q_n$ are degree $n$). We wish to determine if it is possible to find such a pair which are  orthogonal with respect to $\lf$. 
In particular, do there exist  such sequences for which $\lf(P_n  Q_m)=\Lam_n\delta_{n,m}$ with $\Lam_n>0$? These two sequences will then give us a basis and a dual basis for the representation associated with the Al-Salam-Chihara polynomials obtained in \cite{Sasamoto:1999aa,UCHIYAMA:2003aa}.

In order to show such a pair of sequences do indeed exist we    
consider  the   infinite dimensional `bi-moment matrix', $ \bmm $,  whose matrix elements are given by
\begin{equation}\label{binom}
 \bmm_{n,m} =\lf(d^{ n}\,    e^{ m})\,, \qquad n,m\ge 0\,.
\end{equation}
Note, all matrices have rows and columns that are indexed by non-negative integers.

The bi-moment  matrix elements  satisfy a pair of partial difference equations as  given in the following theorem. 
\begin{theorem}\label{bimompascal}
The bi-moment matrix elements, \eqref{binom}, satisfy the recursions 
\begin{subequations}\label{derecs}
\begin{align}
\bmm_{i,j} &= (1-q^{i}) \bmm_{i-1,j-1} + aq^{i} \bmm_{i,j-1} \label{erec}\,,\\
\bmm_{i,j} &= (1-q^{j}) \bmm_{i-1,j-1} + bq^{j} \bmm_{i-1,j} \label{drec}\,,
\end{align}	
\end{subequations}
$i,j>0$ with boundary values $\bmm_{0,j}=a^{j} \textrm{ and } \bmm_{i,0}=b^{i}$, $i,j\ge0$. 
\end{theorem}
Thus the    matrix  looks like 
\begin{equation*}
\bmm=\left(\begin{matrix}
1 & a & a^2 & \hdots \\
b & 1+h_0q  & a(1+h_0q^2 ) & \hdots \\
b^2 & b(1+h_0q^2 )& 1+h_0q (1+q+h_1q^2 ) &  \hdots \\ 
\vdots & \vdots & \vdots & \ddots\\
\end{matrix}\right)
\end{equation*}
where $h_n=abq^n-1$.

\begin{proof} 
The idea of the proof is to  repeatedly use  the commutator \mbox{$de-qed=q'$} from \eqref{eq_cw_a}  to move an  $e$ (resp.\ $d$) to the left (resp.\ right). Thus,
\begin{align}
	d^n e^m &= d^{n-1} (de) e^{m-1}\\
		&= d^{n-1} (q' + qed) e^{m-1}\\
		&=q'd^{n-1}e^{m-1}+qd^{n-1}ede^{m-1} 
\intertext{and}
	d^{n-1}ede^{m-1} &=  q'd^{n-1}e^{m-1} + qd^{n-2}ed^{2}e^{m-1}\,.
\intertext{Commuting the lone $e$ to the left gives}	
	d^{n-1}ede^{m-1} &= \left( (1-q)\sum_{k=1}^{n-1} q^{k-1} \right)  d^{n-1} e^{m-1}   
+ q^{n-1}e d^ne^{m-1}  	
\intertext{and since the sum telescopes   the result is}
	d^n e^m &=(1-q^n) d^{n-1} e^{m-1}   + q^{n}e d^ne^{m-1}\,.
\end{align}
Using \eqref{eq_cw_c} gives \eqref{erec}. Similarly for \eqref{drec}
\end{proof}

As will be shown below, the existence of the BiOPS requires that the determinant of the $(n+1)\times (n+1) $ sub-matrix  
$$\bmm^{(n)}=(\bmm_{i,j})_{0\le i,j\le n}$$ 
be non-zero for all $n\ge 0$. Thus we require the following theorem.
\begin{theorem}\label{detthm}
Let $\bmm^{(n)}=(\bmm_{i,j} )_{0\le  i,j\le n}$ be the  truncated $(n+1)\times(n+1)$ bi-moment matrix whose elements are defined by \thmref{bimompascal}.  Then 
\begin{equation}\label{eq_bidet}
	\det\,\bmm^{(n)}=\prod_{i=1}^n\left(1-q^i\right)^{n+1-i} \left(1-ab\, q^{i-1}\right)^{n+1-i}\,.
\end{equation}
\end{theorem} 
The value of the determinant is a simple consequence of the LDU-decomposition of the bi-moment matrix as given by the following theorem.
\begin{theorem}\label{tm_trirecrel}
The LDU-factorisation of the bi-moment matrix is 
\begin{equation}
	\bmm=\lmm \dmm \umm 
\end{equation}
where the three matrices have elements determined by the following $q$-recurrences. The lower triangular matrix elements satisfy
\begin{subequations}\label{eq_ulrec}
\begin{align}
\lmm_{i,j} &= \lmm_{i-1,j-1} + b q^{j} \lmm_{i-1,j} 
\intertext{with $\lmm_{i,0} = b^{i}$ and  $\lmm_{0,j} = \delta_{0,j}$. The upper triangular matrix elements satisfy}
\umm_{i,j} &= \umm_{i-1,j-1} + a q^{i} \umm_{i,j-1} 
\intertext{with $  \umm_{0,j} = a^{j} \mbox{ and } \umm_{i,0} = \delta_{i,0}$ and the diagonal matrix elements satisfy}
\dmm_{j} &= (1-abq^{j-1})(1-q^j)\dmm_{j-1}
\end{align}	
\end{subequations}
with  $\dmm_0=1$.
\end{theorem}
The proof of the theorem is detailed in  \secref{sec_LDU}. The LDU-decomposition of the bi-moment matrix
is at the centre of the whole calculation. Once the decomposition is obtained most of the other results are straightforward consequences. For the case of  $q=0$ the LDU-decomposition and determinant in the $D$ and $E$ variables has been obtained by Krattenthaler\cite{Krattenthaler:2002aa}. 

We now use the bi-moment matrix to show the existence and uniqueness of the  BiOPS.  For $n,m\ge 0$  require the bi-orthogonality condition
\begin{equation}\label{bi-orthog}
\lf(P_n(d)\,  Q_m(e))=\Lam_m\delta_{n,m}
\end{equation}
where $\Lam_n$ is  a sequence of non-zero normalisation factors  determined by $\lf$ and the monic constraint. 

If this bi-orthogonality is translated into the inner product form of the original matrix product Ansatz, then the equation  is asking the question: Does there exist polynomials  $P_n(\db)$ and $Q_m(\eb)$  in the matrices $\db$ and $\eb$ such that
\begin{equation}\label{eq_maxbiorth}
	\bra{W}P_n(\db)\,Q_m(\eb)\ket{V}=\Lam_m\delta_{n,m}\, 
\end{equation}
for vector  $\ket{V}$ and dual vector $\bra{W}$ defined by 
$$(\db-b\mathbf{1})\ket{V}=0$$ 
and   
$$\bra{W}(\eb-a\mathbf{1})=0\,?$$ 
If so we get sequences of basis vectors $ \ket{\hat{V}_n}_{n\ge0}$  and  their orthonormal  (with respect to $\lf$) duals  $ \bra{\hat{W}_n}_{n\ge0}$,  given by 
\begin{equation}\label{eq_orthobasis}
	\bra{\hat{W}_n}=\bra{W}P_n(\db)\frac{1}{\sqrt{\Lam_n}}\qquad\text{and}\qquad\ket{\hat{V}_n}=\frac{1}{\sqrt{\Lam_n}}Q_n(\eb)\ket{V}\,,
\end{equation}
where $\ket{\hat{V}_0}=\ket{ V }$ and $\bra{\hat{W}_0}=\bra{W }$. We normalise so that $ \bran{W }\ket{ V }=1$. 
From these sequences, and since the identity matrix is $\mathbf{1}=\sum_{n\ge 0}\ket{\hat{V}_n} \bra{\hat{W}_n}$,  we get matrix representations of $\db$ and $\eb$ via 
\begin{equation}\label{eq_bdrybasis}
	\db_{n,m}=\bra{\hat{W}_n} d \ket{\hat{V}_m}\qquad \text{and}\qquad \eb_{n,m}=\bra{\hat{W}_n} e\ket{\hat{V}_m}
\end{equation}
which satisfy  \eqref{eq_cw_a}.  

This procedure for constructing basis vectors (see for example \cite{Lazarescu:2014ab}) is analogous to the quantum oscillator basis set $\ket{n}_{n\ge0}$   constructed by the action of $\mathbf{e}^n$ on a vacuum vector $\ket{0}$ which is defined by $\mathbf{d}\ket{0}=0$, that is, $\ket{n}=\mathbf{e}^n\ket{0}$. The dual vectors are given via the action of $\mathbf{d}^n$ on the dual vacuum $\bra{0}$,   
that is  $\bra{n}\,\prod_{i=1}^n(1-q^i)=\bra{0}\mathbf{d}^n$.
In the BiOPS case the boundary vector $\ket{V}$ plays the role of the vacuum vector and the basis set $\ket{V_n}_{n\ge0}$ is generated by the action of $Q_n(\eb)\ne \eb^n$ on $\ket{V}$ defined by $(\db-b\mathbf{1})\ket{V}=0$. 
The dual vectors $\bra{W_n}$ are similarly related to the action of $P_n(\db)$ on the dual boundary vector $\bra{W}$.


%
%

Returning to the question of the existence of bi-orthogonal polynomials  we  have the following theorem   stating a unique   pair of sequences exists.
\begin{theorem} \label{thm_othog}
 Let   $\set{P_n(d)}_{n\ge0}$ and $\set{Q_n(e)}_{n\ge0}$ be a pair of sequences of monic polynomials   satisfying 
\begin{equation}
\lf(P_n  Q_m)=\Lam_n\delta_{n,m}
\end{equation}
where  the linear functional $\lf$ is defined by  equations \eqref{eq_cw}. 
Then $\set{P_n}_{n\ge0}$ and $\set{Q_n}_{n\ge0}$ exist and are unique with
\begin{equation}\label{eq_lamn}
	\Lam_n=\prod_{i=1}^n (1-abq^{i-1})(1-q^i) \,.
\end{equation}
for $n>0$ and $\Lam_0=1$.
\end{theorem}

\begin{proof}
The existence of $\set{P_n}$ follows by applying Cramer's rule   to the system of linear equations obtained by writing 
\begin{equation}\label{poldef}
P_n(d)=\sum_{k=0}^n a_k^{(n)} d^k  
\end{equation}
 with  $a^{(n)}_n =1$  and for $ k\le n$,
\begin{equation}
\lf(P_n e^k)=\Lam_n \delta_{n,k}\,.
\end{equation}
Since   $  e^k = Q_k(e)+\sum_{\ell=0}^{k-1} c_\ell^{(k)} Q_\ell(e)$  and using equations \eqref{binom} and   \eqref{bi-orthog}  we get the system of equations
\begin{equation}
(a^{(n)}_0,a^{(n)}_1,\dots, a^{(n)}_n)\left(\begin{matrix}
\bmm_{0,0} & \bmm_{0,1} & \hdots  &\bmm_{0,n}\\
\bmm_{1,0} & \bmm_{1,1} & \hdots  & \bmm_{1,n}\\
\vdots & \vdots &\ddots & \vdots\\
\bmm_{n,0} & \bmm_{n,1} &  \hdots  &\bmm_{n,n}
\end{matrix}\right)=
\left( 
 0,
 \dots,
 0,
  \Lam_n \right)\,.
\end{equation}
Since for all $ n\ge 0$ we have from \thmref{detthm} that $\det\,\bmm^{(n)}\ne 0$  and thus  the system has a unique solution given by Cramer's rule 
\begin{subequations}\label{expoly}
\begin{equation}\label{expolyp}
P_n(d)=\frac{1}{\det\,\bmm^{(n-1)}}\det\left(\begin{matrix}
\bmm_{0,0} & \bmm_{0,1} & \hdots  &\bmm_{0,n-1}  &1\\
\bmm_{1,0} & \bmm_{1,1} & \hdots  & \bmm_{1,n-1}& d\\
\vdots & \vdots &\ddots & \vdots & \vdots\\
\bmm_{n-1,0} & \bmm_{n-1,1} &  \hdots  &\bmm_{n-1,n-1}& d^{n-1}\\ 
\bmm_{n,0} & \bmm_{n,1} &  \hdots  &\bmm_{n,n-1}& d^n 
\end{matrix}\right) \,.
\end{equation}
Similarly 
\begin{equation}\label{expolyq} 
Q_n(e)=\frac{1}{\det\,\bmm^{(n-1)}}\det\left(\begin{matrix}
\bmm_{0,0} & \bmm_{0,1} & \hdots & \bmm_{0,n-1} &\bmm_{0,n}\\
\bmm_{1,0} & \bmm_{1,1} & \hdots & \bmm_{1,n-1} & \bmm_{1,n}\\
\vdots & \vdots &\ddots & \vdots & \vdots\\
\bmm_{n-1,0} & \bmm_{n-1,1} &  \hdots &\bmm_{n-1,n-1}  &\bmm_{n-1,n}\\
1 & e &\hdots & e^{n-1} & e^n 
\end{matrix}\right)\,.
\end{equation}	
\end{subequations}
The scalar $\Lam_n$ follows from the monic requirement  which gives
\[
	\Lam_n=\det\bmm^{(n)}/\det\bmm^{(n-1)}\,.
\]
and hence from \eqref{eq_bidet} we get \eqref{eq_lamn}.
\end{proof}


To find the explicit form of the polynomials we need to evaluate the two determinants \eqref{expolyp}  and \eqref{expolyq}. This requires the LDU-decomposition of the two matrices leading to the following lemma.

\begin{theorem}  \label{lem_polrec}
The  pair of sequences of monic polynomials   $\set{P_n(d)}_{n\ge0}$ and $\set{Q_n(e)}_{n\ge0}$   satisfy the recurrence relations 
\begin{subequations}\label{eq_polexpn}
\begin{align}
	P_n(d) &= d^n-\sum_{k=0}^{n-1} \lmm_{n,k} P_k(d)\label{eq_polexpnp}\,,\\
	Q_n(e) &= e^n-\sum_{k=0}^{n-1} Q_k(e)\umm_{k,n} \label{eq_polexpnq}\,,
\end{align}
\end{subequations}
where $\lmm_{n,k}$ and $\umm_{k,n}$ are the matrix elements of the lower triangular $\lmm$ and  upper triangular $\umm$    are given by \eqref{eq_ulrec}.
\end{theorem}

\begin{proof}

The theorem  follows from the LDU-decomposition (detailed in \secref{sec_LDU}) of the bi-moment matrix. This decomposition  reduces \eqref{expoly} to the single determinant forms 
\begin{subequations}
\begin{equation}\label{lduexpolyp}
P_n(d)=\det\left(\begin{matrix}
\lmm_{0,0}		& 0 			&\hdots	& 0				& 1			\\
\lmm_{1,0}		& \lmm_{1,1} 	&\hdots	& 0				& d			\\
\vdots			& \vdots 		&\ddots	& \vdots			& \vdots		\\
\lmm_{n-1,0}	& \lmm_{n-1,1}	&\hdots	& \lmm_{n-1,n-1}	& d^{n-1}	\\
\lmm_{n,0}		& \lmm_{n,1}	&\hdots	& \lmm_{n,n-1}	& d^n 
\end{matrix}\right)
\end{equation}
and
\begin{equation}\label{lduexpolyq} 
Q_n(e)=\det\left(\begin{matrix}
\umm_{0,0}	& \umm_{0,1}	& \hdots	& \umm_{0,n-1}	& \umm_{0,n}		\\
0			& \umm_{1,1}	& \hdots	& \umm_{1,n-1}	& \umm_{1,n}		\\
\vdots		& \vdots		& \ddots	& \vdots			& \vdots			\\
0			& 0				& \hdots	& \umm_{n-1,n-1}	& \umm_{n-1,n}	\\
1			& e				& \hdots	& e^{n-1}		& e^n 
\end{matrix}\right)\,.
\end{equation}
\end{subequations}
Expanding  \eqref{lduexpolyp} using  the bottom row leaves a sub-matrix determinant which reduces down to a $k \times k$ determinant of the same form   as \eqref{lduexpolyp}  but with $n=k$ and hence is   $P_k(d)$. Thus we get \eqref{eq_polexpnp}. Similarly for  \eqref{eq_polexpnq}.
\end{proof}

We now use \eqref{lduexpolyp} to find explicit forms for $P_n$ and $Q_n$.
\begin{theorem}\label{tm_biorthpoys}
The  pair of sequences of monic polynomials $\set{P_n(d)}_{n\ge0}$ and $\set{Q_n(e)}_{n\ge0}$   are given by
\begin{align*}
	P_{n}(d) &= \prod_{k=1}^n (d-bq^{k-1})\,,  
\intertext{and}
	Q_{n}(e) &= \prod_{k=1}^n (e-aq^{k-1})   
\end{align*}
with $P_0=Q_0=1$.
\end{theorem}
 \begin{proof}
The theorem is equivalent to     $P_n$ and $Q_n$ satisfying the first order recurrence relations
\begin{subequations}\label{eq_pqrecrel}
\begin{align}
P_{n+1}(d) &= (d-bq^{n})P_{n}\label{eq_pqrecrelp} \\
Q_{n+1}(e) &= (e-aq^{n})Q_{n}\label{eq_pqrecrelq}  
\end{align}	
\end{subequations}
which we  prove by induction using the recurrence relations \eqref{eq_ulrec} satisfied by the upper and lower triangular matrix elements.
From \thmref{lem_polrec} we get 
\begin{align}
dP_n = d^{n+1} - \sum_{k=0}^{n-1} \lmm_{n,k} dP_k \label{xpn}\,.
\end{align}
The induction assumption is  
\begin{equation}\label{eq_dprec}
dP_{n-1} = P_n + bq^{n-1} P_{n-1}	\,.
\end{equation}
Using  \eqref{eq_dprec}  and \eqref{xpn} gives
\begin{equation}\label{eq_notlan}
dP_n = d^{n+1} - \sum_{k=0}^{n-1} \lmm_{n,k} (P_{k+1} + bq^{k} P_{k}) \,.
\end{equation}
From \thmref{lem_polrec} we have
\begin{align*}
P_{n+1} &= d^{n+1} - \sum_{k=0}^n \lmm_{n+1,k} P_k
\intertext{and using the recurrence relation  for $\lmm$ from \thmref{tm_trirecrel} gives}
P_{n+1} &=   d^{n+1} - \sum_{k=0}^{n-1} \lmm_{n,k} (P_{k+1} + bq^{k} P_{k}) -bq^n \lmm_{n,n}P_n\,. 
\intertext{Using \eqref{eq_notlan} and since $\lmm_{n,n}=1$ we get}
P_{n+1}&= d P_n - bq^nP_n \,.
\end{align*}
Since $n=1$ is true,  by induction, we have shown  \eqref{eq_pqrecrelp}. A similar induction proof gives \eqref{eq_pqrecrelq}.
\end{proof}


\section{Matrix representation in boundary basis}  

In this section we briefly discuss a representation of the linear functional $\lf$ by an inner product using a matrix representation of the tensor algebra. 

The polynomials $Q_n(e)$ generate a basis set $\set{v_n}_{n\ge0}$ for the module  \eqref{eq_mmod} by their action on the boundary monomial element $v_0$ satisfying  $\lf\bigl(u(d-b)v_0\bigr)=0$, that is,   generated by the set of elements $v_n=Q_n(e)v_0$. Denote the  module in this basis by $V_{Q}$. 
Note, equation \eqref{eq_cw_b} shows that in the tensor space $v_0=1$ but is usually denoted $\ket{V}$ when $\lf$ is represented by an inner product.

It is well known that since $V_Q$ is infinite dimensional that its dual space is not spanned by the elements, $v^*_n$ dual to $v_n$ (ie.\ $v^*_n(v_m)=\delta_{n,m}$).  Thus it is not clear \emph{a priori} that all linear functionals $\lf$ that satisfy \eqref{eq_cw} can be expressed as an element in the dual sub-module spanned by $v^*_n$.  
However, for the computational purposes of the ASEP model we only require a non-trivial such linear functional. It turns out to be  sufficient to restrict ourselves to linear functionals in the dual sub-module spanned by $v^*_n$.  Call this dual sub-module $V^*_P$. Thus we seek a linear functional $\lf$ satisfying  \eqref{eq_cw} that exists in the dual sub-module $V^*_P$. 

 \thmref{thm_othog} tells us that given the set $\set{v_n}$ there exists a unique dual set $v^*_n= P_n(d)$. We first find a matrix representation of the quotient module 
 \begin{equation}\label{eq_quomod}
   \mathcal{S}=\mmod/ (d  e - q\, ed -q')  	
 \end{equation}
and then address the question of how to extract $\lf(g)$, $g\in\mathcal{S}$, from the matrix representation of $g$.

 In order to obtain a matrix representation  we need to use   normalised sequences $\{\nP_n\}$, $\{\nQ_n\}$ of the two polynomials sequences. If \eqref{bi-orthog} is replaced by
\begin{equation}
\lf(\nP_n\, \nQ_m)=\delta_{n,m}\,,
\end{equation}
then clearly
\begin{subequations}\label{eq_othonorm}
\begin{align}
\nP_n=&P_n/\sqrt{\Lambda_n}\,,\label{eq_pnorm}\\
\nQ_n=&Q_n/\sqrt{\Lambda_n}\,.\label{eq_pnorm2}
\end{align}
\end{subequations}
gives a  bi-orthonormal pair of polynomial sequences.

The recurrence relations \eqref{eq_pqrecrel} for $P_n$ and $Q_n$  can be used   to compute the following two moments which lead to a matrix representation.  

\begin{theorem} \label{thm_fmom}
Let $P_n$ and $Q_n$ be the polynomials of  \thmref{tm_biorthpoys}. The two first moments
\begin{subequations}
\begin{align}
X_{n,m}&=\lf(P_n\, d\, Q_m), \\
Y_{n,m}&=\lf(P_n \,e\, Q_m), 
\end{align}	
\end{subequations}
for $n,m\ge0$, are given by
\begin{subequations}
\begin{align}\label{fmmat}
X_{n,m}&=  \Lam_{n+1} \delta_{n+1,m}+b q^{n} \Lam_{n}\delta_{n,m}\,,  
 \\
Y_{n,m}&= \Lam_{m+1} \delta_{n,m+1}+a q^{m} \Lam_{m}\delta_{n,m}\, ,  
\end{align}	
\end{subequations}
for $n,m\ge 0$. 
\end{theorem}

The orthonormal versions of the polynomials give rise to a representation of the quotient module \eqref{eq_quomod}.
 \begin{theorem} \label{tm_matprod}
 	 The infinite dimensional matrices $\db$ and $\eb$ with matrix elements
\begin{subequations}\label{eq_matrefdssd}
\begin{align}
	\db_{n,m}&=\lf(\nP_n \,d\, \nQ_m)=X_{n,m}/\sqrt{\Lam_n\Lam_m}\,, \\
	\eb_{n,m}&=\lf(\nP_n\, e \,\nQ_m)=Y_{n,m}/\sqrt{\Lam_n\Lam_m}\,,, 
\end{align}
\end{subequations}
for $n,m\ge0$, give a matrix representation of the quotient module \eqref{eq_quomod}.
\end{theorem}
The theorem is proved by direct verification that the matrices \eqref{eq_matrefdssd} satisfy the quotient relation $\mathbf{d  e} - q\, \mathbf{e d} = q'\mathbf{1}$.  

The matrices \eqref{eq_matrefdssd} have a simple bi-diagonal structure
\begin{subequations}\label{eq_matrep}
\begin{equation}\label{eq_matrepd}
\db=\left(\begin{matrix}
b & \sqrt{c_0} & 0 & \hdots \\
0 & bq & \sqrt{c_1} & \hdots \\
0 & 0 & bq^2 &  \hdots \\ 
\vdots & \vdots & \vdots & \ddots\\
\end{matrix}\right)
\end{equation}
and
\begin{equation}\label{eq_matrepe}
\eb=\left(\begin{matrix}
a & 0 & 0 & \hdots \\
\sqrt{c_0} & aq & 0 & \hdots \\
0 & \sqrt{c_1} & aq^2 &  \hdots \\ 
\vdots & \vdots & \vdots & \ddots\\
\end{matrix}\right)
\end{equation}
\end{subequations}
where $c_n=(1-q^{n+1})(1-ab q^n)$. These are the same matrices obtained by Sasamoto \cite{Sasamoto:1999aa}.



%

The following theorem states the relationship between $\lf$ and the matrix representation.  
 \begin{theorem} 	 \label{tm_lfex}
 Let $\mathbf{m}$ be the matrix representation of an element $m$ in the quotient module \eqref{eq_quomod}. Then
\begin{equation}\label{eq_lmonval}
	\lfp(m)=  \mathbf{m}_{0,0}  	
\end{equation} 
where $\mathbf{m}_{0,0}$ is the $(0,0)$ matrix element of $ \mathbf{m}$.
 \end{theorem}

Equation \eqref{eq_lmonval} is the matrix product representation of $\lf$ conventionally written 
\begin{align}
	\lfp\bigl( {e_1} {e_2} \cdots  {e_k}\bigr)&=\bra{W}  \mathbf{e_1}\mathbf{e_2} \cdots \mathbf{e_k}\ket{V}
	\intertext{where  $\mathbf{e_i}\in\set{\mathbf{d},\mathbf{e}}$. In the basis $\ket{V_n}$ we have }
	\bra{W}=\bra{W_0}=\left( 1, 0, \dots, \right)\,,
	&\qquad\text{and} \qquad  \ket{V}=\ket{V_0}=\left( 1, 0, \dots, \right)^T\,. 
\end{align} 
were $T$ denotes the transpose. 
\begin{proof}[Proof of \thmref{tm_lfex}.]
	Clearly \eqref{eq_lmonval} defines a linear functional from the space of matrices to $\R$. It remains to verify that such a functional satisfies the equations \eqref{eq_cw}. 
 Equation \eqref{eq_cw_a} is satisfied as $\mathbf{d  e }- q\, \mathbf{ed} -q'\mathbf{1}$  is the zero matrix. 
 Equation  \eqref{eq_cw_b}  requires $\lf(u(d-b))=  \mathbf{u(d}- b \mathbf{1})_{0,0}=0$ for any $\mathbf{u}\in\mathcal{S}$,  which trivially verified using the matrix   \eqref{eq_matrepd}. Similarly for  \eqref{eq_cw_c}.
 \end{proof}
Since the matrices  $\db$ and $\eb$ are upper and lower bi-diagonal respectively, their sum is clearly tri-diagonal  and hence related to traditional three term recurrence orthogonal polynomials. 
In this case the tri-diagonal matrix is 
\[
\mathbf{W}_{n,m}=\lf(\nP_n( d + e) \nQ_m)=\frac{X_{n,m}+Y_{n,m}}{\sqrt{\Lam_n\Lam_m}}\,.
\]
Thus the three term recurrence relation of the polynomials  $\set{T_n(x)}_{n\ge0}$, is
\begin{equation}
W_{n,n-1} T_{n-1}+(W_{n,n}-x) T_n + W_{n,n+1} T_{n-1}=0
\end{equation}
with initial values $T_0=1$ and $ T_{-1}=0$. These  are essentially the Al-Salam-Chihara polynomials \cite{gasper90}.

\section{LDU-decomposition of the Bi-Moment Matrix}
\label{sec_LDU} 

In this section we derive the decomposition of the bi-moment matrix into a product of a lower triangular matrix $\lmm$, a diagonal matrix $\dmm$ and an upper triangular matrix $\umm$ as given in \thmref{LDUdecomp}. In order to do this we extend a theorem in \cite{A.R.-Moghaddamfar:2008aa} by extracting the upper and lower matrices. 

We start with the definition of matrices whose elements are given by a recurrence relation as stated in \cite{A.R.-Moghaddamfar:2008aa}. We will refer to such a matrix as a recursively defined matrix.

\begin{definition}
Let $\alpha = (\alpha_i)_{i \geq 0}$, $\beta = (\beta_i)_{i \geq 0}$, $\gamma = (\gamma_i)_{i\geq 0}$, $\mu = (\mu_i)_{i \geq 0}$, $\nu = (\nu_i)_{i \geq 0}$, $\epsilon = (\epsilon_i)_{i \geq 0}$   and   $\lambda = (\lambda_i)_{i \geq -1}$ be given sequences. Let 
\begin{displaymath}
\begin{cases}
\Phi(i,j) = \epsilon_{i-1} \gamma_{j-1} + \nu_{i-1} \mu_{j-1} & \text{ for } i, j \geq 1,\\
\Psi(i,j) = \epsilon_{i-1} \lambda_{j-1} + \nu_{i-1} & \mbox{ for } i \geq 1, j \geq 0,\\
\Omega(i,j) = [\alpha_i - \Psi(i,0)\alpha_{i-1}](\beta_j - \mu_{j-1} \beta_{j-1})& \mbox{ for } i, j \geq 1.\\
\end{cases}
\end{displaymath}
A recursively defined  matrix, is a matrix $A=(a_{i,j})$ of order $n+1$ defined  by  $a_{0,j}=\alpha_0 \beta_j$, $a_{i,0}=\beta_0 \alpha_i$ for $0 \leq i,j \leq n$ and 
\begin{displaymath}
a_{i,j} = \mu_{j-1} a_{i,j-1} + \Phi(i,j) a_{i-1,j-1} + \Psi(i,j) a_{i-1,j} + \Omega(i,j)  
\end{displaymath}
for $1 \leq i,j \leq n$.
\end{definition}
For matrices whose elements satisfy the above definition the following theorem gives the decomposition.
\begin{theorem} \label{LDUdecomp}
The unique LDU-decomposition for a recursively defined matrix is 
\begin{displaymath}
A = \lmm \cdot \dmm \cdot \umm
\end{displaymath}
where $\lmm$ (resp. $\umm$) is a lower (resp. upper) triangular matrix with diagonal entries 1 and $\dmm$ is a diagonal matrix. Also $\lmm=(\lmm_{i,j})_{0 \leq i,j \leq n} \mbox { , }l = (l_{i,j})_{0 \leq i,j \leq n}$ and $D^{(1)}$ is a diagonal matrix with diagonal entries $(D^{(1)}_{i})_{0 \leq i \leq n}$ such that  
\begin{align*}
l &= \lmm \cdot D^{(1)} \\
\intertext{and}
D^{(1)}_{i} &= \alpha_0 \prod^{n-1}_{k=1} \epsilon_k
\end{align*}
where $l_{i,0}=\alpha_i$, $l_{0,j}= \delta_{0,j} $ and 
\begin{displaymath}
l_{i,j} = \epsilon_{i-1} l_{i-1,j-1} + \Psi(i,j) l_{i-1,j} 
\end{displaymath}
and where $\umm=(\umm_{i,j})_{0 \leq i,j \leq n} \mbox { , }u = (u_{i,j})_{0 \leq i,j \leq n}$ and $D^{(2)}$ is a diagonal matrix with diagonal entries $(D^{(2)}_{j})_{0 \leq j \leq n}$ such that 
\begin{align*}
 u =& D^{(2)} \cdot \umm \\
 \intertext{and}
D^{(2)}_{j} =& \sum_{k=1}^{j}\Big\lbrace \beta_{k-1} (\gamma_{k-1} + \mu_{k-1} \lambda_{k-2})\prod^{k-2}_{r=0} (\lambda_{j-1} - \lambda_{r-1})\prod^{j}_{s=k+1} (\gamma_{s-1} +\mu_{s-1} \lambda_{j-1})\Big\rbrace\\  &+ \beta_j \prod^{j-1}_{t=0} (\lambda_{j-1} - \lambda_{t-1})
\end{align*}
with $u_{0,j}=\beta_j \mbox{ and } u_{i,0}=\delta_{i,0}$  and 
\begin{displaymath}
u_{i,j} =\mu_{j-1} u_{i,j-1} + (\gamma_{j-1}+\mu_{j-1}\lambda_{i-2}) u_{i-1,j-1} + (\lambda_{j-1}-\lambda_{i-2}) u_{i-1,j}
\end{displaymath}  
and the diagonal matrix $\dmm$ has diagonal elements $\dmm_i$ such that 
\begin{displaymath}
\dmm_i = D_i^{(1)}D_i^{(2)}\,.
\end{displaymath}
\end{theorem}

This is a modified version of the main theorem in \cite{A.R.-Moghaddamfar:2008aa}. There the theorem states the determinant of a recursively defined matrix and proves the result using a LU-decomposition. \thmref{LDUdecomp} converts the LU-decomposition from the proof into the unique LDU-decomposition.

\begin{proof}[Proof of \thmref{tm_trirecrel}]
From the recursion in \eqref{erec} we get that the bi-moment matrix is a recursive matrix with
\begin{displaymath}
\mu_j = 0, \hspace{.3cm} \nu_i = 0, \hspace{.3cm} \epsilon_i = 1, \hspace{.3cm} \gamma_j = 1- q^{j+1}, \hspace{.3cm} \lambda_j = bq^{j+1}, \hspace{.3cm} \alpha_i = b^{i}, \hspace{.3cm} \beta_j = a^{j}
\end{displaymath}
therefore by \thmref{LDUdecomp} we get
\begin{align*}
\lmm_{i,j} &=  \lmm_{i-1,j-1} + bq^{j} \lmm_{i-1,j} 
\intertext{$\mbox { where } \lmm_{i,0} = b^{i} \mbox{ and } \lmm_{0,j} = \delta_{0,j}$ with }
\dmm_{j} &= \sum_{k=1}^{j+1} (ab)^{k-1} \prod_{r=0}^{k-2}(q^{j}-q^r) \prod_{s=k}^{j} (1-q^s)\\ 
\intertext{which can be shown to satisfy}
\dmm_{j} &= (1-abq^{j-1})(1-q^j)\dmm_{j-1}\,.
\end{align*}
To get the upper triangular matrix, we will instead find the lower triangular matrix of the transpose of the bi-moment matrix. From the recursion in \eqref{drec} , we get that the transpose of the bi-moment matrix is a recursive matrix with
\begin{displaymath}
\mu_j = 0, \hspace{.3cm} \nu_i = 0, \hspace{.3cm} \epsilon_i = 1, \hspace{.3cm} \gamma_j = 1- q^{j+1}, \hspace{.3cm} \lambda_j = aq^{j+1}, \hspace{.3cm} \alpha_i = a^{i}, \hspace{.3cm} \beta_j = b^{j}
\end{displaymath}
therefore by \thmref{LDUdecomp} 
\begin{align*}
\umm_{i,j}^\intercal =  \umm^\intercal_{i-1,j-1} + aq^{j} \umm^\intercal_{i-1,j} 
\end{align*}
where $\umm^\intercal_{0,j} = a^{j}$ and $\umm^\intercal_{i,0} = \delta_{i,0}$.
Taking the transpose of this matrix gives the required result. 
\end{proof}

\section{Concluding Remarks}

We have shown that the representation associated with the Al-Salam-Chihara polynomials obtained by Sasamoto \cite{Sasamoto:1999aa} is a matrix representation of the quotient module  $\mmod/(de-qed-q' )$ with respect to a basis $\ket{V_n}=Q_n\ket{V}$ generated by the boundary vector $\ket{V}$ via the action of the polynomial sequence $\set{Q_n}$. 
The vectors, $\bra{W_n}=\bra{W}P_n$, dual to $\ket{V_n}$ are generated by the dual boundary vector $\bra{W}$ through the action of the polynomials $\set{ P_n }$. 
The two sequences $\set{P_n}$ and $\set{Q_n}$ are bi-orthogonal with respect to the linear functional $\lf$ defined by the equations \eqref{eq_cw}, that is $\lf(P_nQ_m)=\Lam_n\delta_{n,m}$. 
Using the bi-moment matrix \eqref{binom} we showed that the two  bi-orthogonal sequences exist and are unique. Through the LDU-decomposition of the bi-moment matrix it is possible to find explicit forms for the bi-orthogonal sequences in the case of the three parameter model.


It would also be of interest to compute the five parameter versions of $P_n$ and $Q_n$ which would presumably be associated with the same Askey-Wilson polynomials \cite{Askey:1975yq} obtained in \cite{UCHIYAMA:2003aa,corteel:2010rt}. Preliminary work shows the five parameter generalisation of the two $q$-recurrence relations, \eqref{derecs},  for the bi-moment matrix are straightforward to derive but that the LDU-decomposition of the resulting matrix is significantly more complicated.

Finally, what about the combinatorics of this formalism? The connection between classical orthogonal polynomials and the combinatorics of lattice paths is well established \cite{Viennot1985ah,Flajolet1980rr} as is the combinatorics of the ASEP model \cite{Blythe:2007uq}. Clearly the bi-diagonal structure of the $\mathbf{d}$ and $\mathbf{e}$ matrices connect to binomial lattice paths (aka.\ fully directed paths) and the tridiagonal matrix $W_{n,m}=\lf(P_n d e Q_m)$ to Motzkin paths. However, in this instance there is no Hankel matrix of moments -- it is replaced by the bi-moment matrix.

%
\section{Acknowledgement} We would like to   thank the Australian Research Council (ARC) and the Centre of Excellence for Mathematics and Statistics of Complex Systems (MASCOS) for financial support. I would also like to thank  the referees for their useful comments.

\newpage
\bibliographystyle{unsrt}  
\bibliography{_bibTek_Biorthog}

\end{document}